\documentclass[12pt]{article}

\usepackage{graphicx}
\usepackage{float}
\usepackage{amsmath}
\usepackage{amssymb}
\usepackage{amsthm}
\usepackage{comment}
\usepackage{xcolor}
\usepackage{hyperref}
\hypersetup{colorlinks = false}
\usepackage[numbers,sort&compress]{natbib}
\usepackage[margin=0.75in]{geometry}
\usepackage{complexity}

\newcommand{\be}{\begin{equation}}
\newcommand{\ee}{\end{equation}}
\newcommand{\ba}{\begin{array}}
\newcommand{\ea}{\end{array}}
\newcommand{\bea}{\begin{eqnarray}}
\newcommand{\eea}{\end{eqnarray}}

\newcommand{\ra}{\rangle}
\newcommand{\la}{\langle}

\newcommand{\CC}{\mathbb{C}}
\newcommand{\DD}{\mathbb{D}}
\newcommand{\RR}{\mathbb{R}}

\newcommand{\ket}[1]{|#1\rangle}

\newcommand{\odisk}[1]{{\mathbb{D}_{{#1}}}}

\newcommand{\multi}[1]{\mathbf{{#1}}}
\newcommand{\rr}{\multi{r}}

\newtheorem{dfn}{Definition}

\newtheorem{lemma}{Lemma}
\newtheorem{fact}{Fact}

\newtheorem{corol}{Corollary}

\newtheorem{theorem}{Theorem}

\newcommand{\footremember}[2]{%
    \footnote{#2}
    \newcounter{#1}
    \setcounter{#1}{\value{footnote}}%
}
\newcommand{\footrecall}[1]{%
    \footnotemark[\value{#1}]%
}

\title{Efficient quantum algorithm for Heisenberg spin systems}

\author{Sergey Bravyi\footremember{ibm}{IBM Quantum, IBM T.J. Watson Research Center, Yorktown Heights NY 10598}
\and
David Gosset\footremember{iqc}{Department of Combinatorics and Optimization and Institute for Quantum Computing, University of Waterloo} \footremember{perimeter}{Perimeter Institute for Theoretical Physics}
\and
Yinchen Liu\footrecall{iqc} 
\and
Benjamin Wong \footrecall{iqc}
}

\date{}

\begin{document}
\maketitle
\begin{abstract}
We consider a broad family of Heisenberg-type quantum spin systems that were studied by Suzuki and Fisher in 1971. This family includes, for example, the Heisenberg antiferromagnet on any bipartite graph. The ground and Gibbs states of these models are \textit{Lee-Yang tensors} with radius $1$: they are associated with multilinear polynomials that possess an extraordinary zero-freeness property inside the unit polydisk in the complex plane. For each Hamiltonian in this family we show that the spectral gap between the first-excited and ground-state energies is lower bounded by $2\mu$, where $\mu$ is the magnetic field along the $Z$ direction. Using this result we obtain an efficient quantum adiabatic algorithm for the ground energy of any model in this family. The proof is based on a new inequality that relates the spectral gap of a positive semidefinite operator to its Lee-Yang radius---a quantitative strengthening of prior work of the authors that may find applications elsewhere.
\end{abstract}

Quantum many-body systems lie at the heart of physics, chemistry, and materials science, yet the computational difficulty of predicting their properties is a major obstacle to progress in these areas. This difficulty cannot be avoided in the worst case: computing the ground state energy is intractable, even for classical spin systems \cite{barahona1982computational}. More generally, preparing the ground state of a local quantum Hamiltonian is a QMA-hard quantum optimization problem~\cite{kitaev2002quantumcomp} and very unlikely to admit efficient classical or quantum algorithms in the worst case.

On the other hand there are a handful of special cases where fast classical algorithms have been found. These include free-fermion solvable models (e.g. \cite{kitaev2006anyons,fendley2019free}), those which are amenable to polynomial interpolation/series expansion, including those at sufficiently high temperature~\cite{harrow2020classical,bakshi2024high}, and certain local Hamiltonians which admit rigorous and efficient classical Monte Carlo algorithms~\cite{jerrum1993polynomial,bravyi2014monte,bravyi2016polynomial}. An actively debated question is whether quantum computers can achieve an exponential advantage over classical algorithms for ground state problems of practical relevance to chemistry and materials science~\cite{lee2023evaluating,yoshioka2024hunting}. This prospect has spurred recent advances in quantum Gibbs sampling \cite{chen2023efficient,gilyen2024quantum}, a technique that may ultimately provide strong candidates along these lines.

 In this work we show that the ground energy of any system from a broad family of ferromagnetic Heisenberg-type models can be computed efficiently via quantum adiabatic evolution. We can also efficiently prepare a state with energy that is inverse-polynomially close to the ground energy. We emphasize that prior to our work, there was no known efficient algorithm for this task except in certain restricted special cases \cite{bravyi2016polynomial}. Moreover, the family of Hamiltonians we consider includes models with a sign problem  that are not expected to be accessible using classical Monte Carlo techniques. In this sense our algorithmic results provide a fairly wide family of candidates for quantum advantage. It remains an open question whether or not these ground energy problems can be solved by fast classical  algorithms based on other techniques such as polynomial interpolation \cite{harrow2020classical}. Another question that we leave for future work is whether techniques such as quantum Gibbs sampling can efficiently compute thermal properties of these systems at a given nonzero temperature.
 
To describe our results, let $G=(V,E)$ be a graph with $n$ vertices.
We place a qubit at each vertex of $G$ and  consider an anisotropic Heisenberg model with external magnetic fields described  by a Hamiltonian
\be
\label{SFmodel1}
H=-\sum_{(i,j)\in E} H_{i,j} - \sum_{i\in V}\left(\mu^x_i X_i + \mu^y_i Y_i + \mu^z_i Z_i\right),
\ee
where
\be
\label{SFmodel2}
H_{i,j} =J^z_{i,j} Z_i Z_j + J^{xx}_{i,j} X_i X_j + J^{yy}_{i,j} Y_i Y_j + J^{xy}_{i,j} X_i Y_j + J^{yx}_{i,j} Y_i X_j.
\ee
Here $J$'s and $\mu$'s are real coefficients. Following the work of 
Suzuki and Fisher~\cite{sf71}, consider a special class of such Hamiltonians. 
\begin{dfn}
An $n$-qubit Hermitian operator $H$ is called a Suzuki-Fisher Hamiltonian 
if it has the form defined in Eqs.~(\ref{SFmodel1},\ref{SFmodel2}) 
where all $Z_iZ_j$ terms  are ferromagnetic and dominant such that
\be
\label{eq:ferronorm}
\left\| \left( \ba{cc}
J^{xx}_{i,j} & J^{xy}_{i,j} \\
J^{yx}_{i,j} & J^{yy}_{i,j} \\
\ea\right) \right\|\le J^z_{i,j}
\ee
for all $(i,j)\in E$ and all magnetic $Z$-fields are non-negative such that $\mu^z_i\ge0$ for all $i\in V$. Note that if $J_{ij}^{xy}=J_{i,j}^{yx}=0$ for all $i,j$ then Eq.~\eqref{eq:ferronorm} simplifies to $J^{z}_{i,j} \ge \max\{ |J^{xx}_{i,j}|, |J^{yy}_{i,j}|\}$.
Let $\mathsf{SF}(n)$ be the set of all $n$-qubit Suzuki-Fisher Hamiltonians.
\end{dfn}
The class $\mathsf{SF}(n)$ includes, up to additive constants, EPR Hamiltonians~\cite{king2023improved,ju2025improved}   in which every two-qubit interaction energetically favors the EPR state $(|00\ra+|11\ra)/\sqrt{2}$, as well as their phase-shifted variants~\cite{wong2026lee}.
It also includes, up to a local basis change, 
Heisenberg {\em anti-ferromagnetic} model on any bipartite graph~\cite{king2023improved},
e.g. the 2D or 3D cubic lattice.
Our main result is the following. 
\begin{theorem}
\label{thm:main}
Let $E_1\le E_2\le \cdots \le E_{2^n}$ be the eigenvalues of a Suzuki-Fisher Hamiltonian $H$. Then 
\be
\label{gap_lower}
E_2 - E_1 \ge  2\mu,
\ee
where $\mu=
\min_{i\in V} \mu^z_i$ is the smallest $Z$-field.
\end{theorem}
In particular, if all $Z$-fields $\mu^z_i$ 
are positive then the ground state of $H$ is non-degenerate
and  separated from excited states by an energy gap of at least $2\mu$. The lower bound in Eq.~(\ref{gap_lower})  is tight since it is saturated
for a pure $Z$-field Hamiltonian, e.g. $H=-\mu\sum_{i\in V} Z_i$. 

Theorem~\ref{thm:main} immediately gives an efficient quantum algorithm for approximating the ground energy
of any Suzuki-Fisher Hamiltonian on an arbitrary graph\footnote{We also note that it resolves Conjecture 3 from~\cite{wong2026lee} up to a constant factor.}.
\begin{corol}
\label{cor:qalg}
There exists a quantum algorithm with runtime 
$poly(n,J,\epsilon^{-1})$
that takes as input a Hamiltonian $H\in \mathsf{SF}(n)$ such that all Pauli coefficients of $H$ have magnitude at most $J$
and estimates the ground energy of $H$ within additive error $\epsilon$.
\end{corol}
\begin{proof}
Let $\delta=\epsilon/(2n)$ and consider an adiabatic path
\[
H(s)= sH - (1+\delta -s)\sum_{i\in V} Z_i, \qquad s\in [0,1].
\]
Let $E_1(s)$ be the ground energy  of $H(s)$.
We have $H(s)\in \mathsf{SF}(n)$ for all $s \in [0,1]$
since the class $\mathsf{SF}(n)$ is closed under 
multiplication by a positive scalar and addition of a positive $Z$-field.
The Hamiltonian $H(s)$ has positive $Z$-fields of magnitude at least $\mu(s)\ge 1+\delta-s\ge \delta$.
By Theorem~\ref{thm:main}, the minimum energy gap along the chosen adiabatic path is at least $2\delta$.
Since $H(0)$ has an easy-to-prepare ground state $|0^n\ra$,
quantum adiabatic evolution algorithm  followed by measuring the energy of $H(1)$
outputs  an estimate $E_{QA}$
satisfying $|E_{QA}-E_1(1)|\le \epsilon/2$
 in time $poly(n,J,\epsilon^{-1},\delta^{-1})$, see~\cite{benseny2021adiabatic} and references therein. 
By Weyl's inequality, $|E_1-E_1(1)|\le \|H-H(1)\| =\delta n\le \epsilon/2$. 
Thus $|E_1-E_{QA}|\le \epsilon$.
\end{proof}

The problem of computing the ground energy of the antiferromagnetic Heisenberg model (also known as Quantum MaxCut) has been studied extensively ~\cite{gharibian2019apx,parekh2021laserre,king2023apxqmc,lee2022optimizingquantumcircuitparameters,lee2024improvedquantummaxcut,jorquera2025monogamyentanglementboundsimproved,apte2025improvedalgorithmsquantummaxcut,apte2025conjecturedbounds2localhamiltonians}. Its complexity on \textit{bipartite} graphs was identified as a challenge problem in \cite{gharibian2024guest}, see recent developments in Refs.~\cite{king2023apxqmc,jorquera2025monogamyentanglementboundsimproved,ju2025improved,apte20250,apte2025conjecturedbounds2localhamiltonians,apte2025improvedalgorithmsquantummaxcut}.
As a special case of Corollary \ref{cor:qalg}, we obtain a polynomial-time quantum algorithm for Quantum MaxCut on bipartite graphs. The corresponding Hamiltonian has the form 
\[
H_{QMC} = -\sum_{(i,j)\in E} J_{ij}  |\Psi^-\ra\la \Psi^-|_{i,j}
\]
where $|\Psi^-\ra\la \Psi^-|$   is the projector onto the singlet state $|\Psi^-\ra = (|01\ra-|10\ra)/\sqrt{2}$
and $J_{ij}>0$ are positive edge weights.
Suppose $V=AB$ such that every graph edge connects $A$ and $B$. Let $Y_A=\prod_{i\in A} Y_i$
and $H=Y_A H_{QMC} Y_A$. One can easily check that $H$ is a Suzuki-Fisher Hamiltonian (up to an overall energy shift).
Furthermore, 
eigenvalues of $H$ coincide with those of $H_{QMC}$. Thus we obtain the following.
\begin{corol}
\label{corol:2}
There exists a quantum algorithm with the runtime $poly(n,J,\epsilon^{-1})$ approximating the ground energy of the Quantum 
MaxCut Hamiltonian  on any bipartite graph with $n$ vertices within additive error $\epsilon$.
Here $J=\max_{(i,j)\in E} J_{ij}$ is the maximum edge weight. 
\end{corol}

In Section~\ref{sec:relative} we prove Theorem~\ref{thm:main} using the machinery of Lee-Yang tensors introduced in~\cite{wong2026lee} and reviewed below.

\section{Lee-Yang tensors \label{sec:lyt}}

Suppose $\psi \in (\CC^2)^{\otimes n}$ is  a complex tensor with $n$ binary indices
and entries $\psi_x$ labelled by bit strings $x\in \{0,1\}^n$. 
Define a generating polynomial $f_\psi \, : \, \CC^n \to \CC$ as 
\be\label{eq:lee-yang-poly}
f_\psi(z) = \sum_{x \in \{0,1\}^n} \psi_x z^x, \qquad z^x= \prod_{j=1}^n z_j^{x_j}.
\ee
Here $z=(z_1,\ldots,z_n)$ is a vector of complex variables.
Equivalently, $f_\psi(z)$ is the inner product between an $n$-qubit state $|\psi\ra=\sum_{x\in \{0,1\}^n}\; \psi_x |x\ra$ and a tensor product of single-qubit states $|0\ra + z_a^*|1\ra$
with $a=1,\ldots,n$. In the special case $n=0$ the tensor $\psi \in \CC$ is a scalar and we set $f_\psi=\psi\in \CC$.
Suppose  $r>0$ is a real number. Let
\[
\odisk{r} = \{z \in \CC\, : \, |z|<r\} 
\]
be the open zero-centered disk of radius $r$.
 If $\rr =(r_1,\ldots,r_n)\in \RR_{>0}^n$ is a tuple of radii, let
\[
\odisk{\rr} = \odisk{r_1} \times \cdots \times  \odisk{r_n}
\]
be the corresponding polydisk in $\CC^n$.
\begin{dfn}
We say that $\psi \in (\CC^2)^{\otimes n}$ is a Lee-Yang tensor with radius   $\rr \in \RR_{>0}^n$ if the generating polynomial of $\psi$ has no zeros
in the open zero-centered  polydisk  of radius $\rr$, that is,
\be
f_\psi(z)\ne 0 \quad \mbox{for all $z\in \odisk{\rr}$}.
\ee
A scalar $\psi\in\CC$ is a Lee-Yang tensor iff $\psi\ne 0$. 
\end{dfn}
Let $LY_n(\rr)$ be the set of all Lee-Yang tensors $\psi\in (\CC^2)^{\otimes n}$  with radius $\rr\in \RR_{>0}^n$.  
It follows directly from the definition that $LY_n(\rr)$   is closed under multiplication by nonzero scalars.
It is also closed
under tensor products in the sense
that $\psi \in LY_n(\rr)$ and $\psi'\in LY_{m}(\rr')$ imply $\psi \otimes \psi' \in LY_{n+m}((\rr,\rr'))$.
The following facts are  rephrasings of results from Refs.~\cite{asano1970theorems,ruelle2010characterization}.
\begin{fact}[\bf Closure under contraction]
\label{fact:contraction}
Consider any tensor $\psi \in LY_n(\rr)$ with $n\ge 2$. Let $\phi\in (\CC^2)^{\otimes (n-2)}$ be a tensor obtained from $\psi$ by contracting some pair of indices $i<j$
such that $r_i r_j>1$. Then $\phi \in LY_{n-2}(\rr')$ where $\rr'$ is obtained from $\rr$ by deleting the components $r_i$ and $r_j$. 
\end{fact}
\begin{fact}[\textbf{Closure under limits}]
Consider a sequence of tensors $\{\psi_j\}_{j\geq 1}$ such that $\psi_j\in LY_n(\rr)$ for each $j\geq 1$. If the sequence converges to a nonzero tensor $\psi$, then $\psi\in LY_n(\rr)$.
\label{fact:seq}
\end{fact}
Given a real number $r>0$, let
$LY(r)$ be the set of  all Lee-Yang tensors  with   radius at least $r$
for each variable:
\[
LY(r)=\bigcup_{n\ge 0} \; LY_n((r,r,\ldots,r)).
\]
Suppose $A$ is an $n$-qubit operator. The  flattening of $A$ defines a tensor $\psi  \in (\CC^2)^{\otimes 2n}$
with  entries $\psi_{xy} = \la x|A|y\ra$, where $x,y\in \{0,1\}^n$.
We say that $A\in LY(r)$ if the flattening of $A$ is contained in $LY(r)$.
Finally, 
Suzuki and Fisher proved the following~\cite{sf71}.
\begin{fact}[\bf Gibbs states]
Consider any Hamiltonian $H\in \mathsf{SF}(n)$.  Then the Gibbs state $e^{-\beta H} /\mathrm{Tr}(e^{-\beta H})$
is contained in $LY(1)$ for all $\beta \ge 0$.
\label{fact:sf71}
\end{fact}
A proof of all the above facts can be found in~\cite{wong2026lee}.

\noindent
{\em Comment:}
Here we use a slightly more general definition of the Suzuki-Fisher models than the one in \cite{wong2026lee} which specialized to the case $J_{i,j}^{xy}=J_{i,j}^{yx}=0$ for simplicity. However, Fact 3 for the more general definition used here follows directly from its specialization to the case $J_{i,j}^{xy}=J_{i,j}^{yx}=0$, stated as Fact 1 in \cite{wong2026lee}. Indeed, Eq.~(\ref{eq:ferronorm}) ensures that 
for every edge $(i,j)\in E$ 
one can choose a pair of edge-dependent single-qubit $Z$-rotations $R_i$, $R_j$ such that 
$H_{i,j} = (R_i R_j) H_{i,j}'(R_i R_j)^\dag$, where $H'_{i,j}$ is of the form Eq.~\eqref{SFmodel2} with $J_{i,j}^{xy}=J_{i,j}^{yx}=0$. Since single-qubit $Z$-rotations preserve $LY(1)$, the arguments used in~\cite{sf71,wong2026lee} to prove the inclusion $e^{-\beta H}\in LY(1)$ are unchanged.

\section{Spectral gap from Lee-Yang radius}
\label{sec:relative}

The gap bound stated in Theorem~\ref{thm:main} is a  corollary of the following relationship between the relative spectral gap of a positive semidefinite operator and its Lee-Yang radius. This is a quantitative strengthening of Theorem 4 from Ref.~\cite{wong2026lee}.
\begin{theorem}[\bf Relative spectral gap]
\label{thm:relative_gap}
Suppose $r>1$ and $K\in LY(r)$ is a positive semidefinite $n$-qubit
operator.  Let $\lambda_1\ge \lambda_2\ge \cdots\ge \lambda_{2^n}$ be the eigenvalues of $K$.
Then the largest eigenvalue $\lambda_1$ is positive and non-degenerate. Furthermore,
\be
\label{relative_gap_bound}
\frac{\lambda_2}{\lambda_1} \le \frac1{r^2}.
\ee
\end{theorem}
\noindent
As a consequence, the relative spectral gap of $K$ satisfies $(\lambda_1-\lambda_2)/\lambda_1 \ge 1-r^{-2}$.
Note the bound Eq.~(\ref{relative_gap_bound}) is tight. For example, suppose $n=1$ and 
$K=\lambda_1|0\ra\la 0| + \lambda_2|1\ra\la 1|$. The generating polynomial of the flattened tensor is
$f_K(z)= \lambda_1 ( 1 + \alpha z_1 z_2)$, where $\alpha=\lambda_2/\lambda_1$.
Choosing $z_1=z_2=i/\sqrt{\alpha}$ gives $f_K(z)=0$. Conversely, if $|z_1|,|z_2|<1/\sqrt{\alpha}$  then
$f_K(z)\ne 0$. Thus $K$ has Lee-Yang radius $r=1/\sqrt{\alpha}$. Equivalently, $\alpha = r^{-2}$,
saturating the bound Eq.~(\ref{relative_gap_bound}).

Let us first show that Theorem~\ref{thm:main} follows rather directly from Theorem \ref{thm:relative_gap}. 

\begin{proof}[\bf{Proof of Theorem \ref{thm:main}}]

Assume wlog that $\mu>0$ since otherwise there is nothing to prove.
Define a magnetization operator $M=\mu\sum_{i=1}^nZ_i$  and let $H_0=H+M$. Since the $Z$-fields of $H_0$ are
$\mu^z_i-\mu\ge0$, one has $H_0\in \mathsf{SF}(n)$ and
$e^{-\beta H_0}\in LY(1)$ for all $\beta \ge0$ due to Fact~\ref{fact:sf71}.
For each $\beta>0$ define an operator
\[
A_\beta  = e^{\beta M/2}\,e^{-\beta H_0}\,e^{\beta M/2}.
\]
Clearly, $A_\beta$ is a hermitian and positive definite operator. 
From $e^{-\beta H_0} \in LY(1)$ 
one gets
\(
A_\beta \in LY(e^{\beta \mu})
\)
since the generating polynomial of $A_\beta$ is obtained from the one of $e^{-\beta H_0}$
by rescaling each variable by $e^{-\beta \mu}$.
Then  the relative spectral gap theorem (Theorem~\ref{thm:relative_gap}) applies to $A_\beta$ with $r=e^{\beta \mu}>1$.
Let  $\lambda_1(\beta)>\lambda_2(\beta)\ge \cdots \ge \lambda_{2^n}(\beta)$ be the eigenvalues of $A_\beta$.
Theorem~\ref{thm:relative_gap} gives
\[
\frac{\lambda_2(\beta)}{\lambda_1(\beta)}\;\le\;e^{-2\beta \mu}.
\]
Equivalently,
\be
\label{eq:ratio}
-\frac1\beta \log\frac{\lambda_2(\beta)}{\lambda_1(\beta )}\;\ge\;2\mu
\ee
for all $\beta>0$.
On the other hand, expanding the three exponentials in $A_\beta$ one gets $
A_\beta =I-\beta H+O(\beta^2)$. By Weyl's inequality, 
\[
\lambda_1(\beta)=1-\beta E_1+O(\beta^2)\quad \mbox{and} \quad
\lambda_2(\beta)=1-\beta E_2+O(\beta^2).
\]
Therefore
$-\frac1\beta \log\frac{\lambda_2(\beta)}{\lambda_1(\beta)}=E_2-E_1+O(\beta)$.
Combining with Eq.~(\ref{eq:ratio}) and letting $\beta\to0$ gives
$E_2-E_1\ge2\mu$, as claimed in Theorem~\ref{thm:main}.
 Note that the above proof applies to any Hamiltonian $H$ such that $H_0=H+M$ obeys
$e^{-\beta H_0} \in LY(1)$ for all small enough $\beta\ge 0$.
\end{proof}

In the rest of this section we prove Theorem~\ref{thm:relative_gap}.
We begin with an informal sketch of the proof. Suppose for simplicity that $\lambda_1>\lambda_2>0$.
Define the spectral ratio $\alpha=\lambda_2/\lambda_1$ and the ratio function $g(z)=f_{\psi_2}(z)/f_{\psi_1}(z)$,
where $f_{\psi_j}(z)$ is the generating polynomial of the eigenvector $\psi_j$ such that $K\psi_j = \lambda_j \psi_j$.
(In the $n=1$ example considered above, $|\psi_1\ra=|0\ra$ and $|\psi_2\ra = |1\ra$.
Thus $f_{\psi_1}(z)=1$ is the constant function, $f_{\psi_2}(z)=z$, and thus $g(z)=z$.)
We will show that $\psi_1\in LY(r)$ so that $f_{\psi_1}(z)$ is non-vanishing on the polydisk $\DD_r^n$
and $g(z)$ is holomorphic on $\DD_r^n$.
Moreover, we show that 
 $g$ has a {\em squeezing property}:  after multiplication by $\alpha$, any nonzero value of $g(z)$ on a  polydisk
$\DD_t^n$
of radius $t\approx r$  must reappear on a smaller polydisk $\DD_s^n$  of radius $s\approx 1/r$. 
We then consider the images $g(\DD_s^n) \subseteq g(\DD_t^n)$ in the complex plane. 
Let $\Delta_s$ and $\Delta_t$ be the diameters of $g(\DD_s^n)$ and $g(\DD_t^n)$.
The squeezing property  gives a bound $\alpha \Delta_t \le \Delta_s$.
We use the Schwarz lemma from complex analysis to  get a complementary bound $\Delta_s \le (s/t) \Delta_t$.
Combining the two bounds gives $\alpha \le s/t$. Taking the limits $s\to 1/r$ (from above) and $t\to r$ (from below) proves
$\alpha\le r^{-2}$, which is the desired bound. 
The proof for the general case follows the same strategy except that we explicitly construct the maximum eigenvector
$\psi_1 \in LY(r)$ using the power method with the initial state $|0^n\ra$. Our argument then implies that 
the maximum eigenvalue $\lambda_1$ is non-degenerate.

\begin{proof}[\bf Proof of Theorem~\ref{thm:relative_gap}]
Below we use notation
\[
\DD_r^n=\{ z\in \CC^n \, : \, \|z\|_\infty<r\} \quad \mbox{and} \quad
\overline{\DD_r^n}=\{ z\in \CC^n \, : \, \|z\|_\infty\le r\}
\]
for the open and closed polydisks of radius $r$.
Pick an arbitrary orthonormal  eigenbasis of $K$ such that 
\[
K=\sum_{j=1}^{2^n} \lambda_j |\psi_j\ra\la \psi_j|
\]
and let 
\[
\lambda_*=\max\{\lambda_j \, : \, \la 0^n|\psi_j\ra \ne 0\}.
\]
We claim that $\lambda_*$ is well-defined and positive.
Indeed, the generating polynomial of $K$ evaluated at the all-zero vector
coincides with $\la 0^n|K|0^n\ra$. From  $K\in LY(r)$ one gets $\la 0^n|K|0^n\ra\ne 0$.
Since $K$ is positive semidefinite, one gets $\la 0^n|K|0^n\ra>0$. Thus 
at least one eigenvector of $K$ with a positive eigenvalue must have a non-zero overlap with $|0^n\ra$
proving $\lambda_*>0$. (We do not assume at this point that
$\lambda_*$ is the largest eigenvalue of $K$ but this will follow.)

Since $LY(r)$ is closed under tensor contractions, see Fact~\ref{fact:contraction},
one has $K^m\in LY(r)$ for every $m\ge1$.
Indeed, one can obtain $K^m$ from $K^{\otimes m}$ by contracting a subset
of tensor indices such that each contracted pair of indices
 has radii with product $r^2>1$.  
 Note that $|0^n\ra \in LY(R)$ for any $R>0$.  
 Choose $R$ such that $rR>1$.
Thus $K^m |0^n\ra \in LY(r)$ for any $m$ since $K^m|0^n\ra$ is obtained from
 $K^m \otimes |0^n\ra$ by contracting pairs of indices with the product of radii $rR>1$.
 Let 
\be
\label{max_eigenvector}
|\psi\ra = \lim_{m\to \infty}\frac{K^m |0^n\ra}{ \lambda_*^{m}}. 
\ee
Equivalently, $|\psi\ra = \Pi_* |0^n\ra$, where $\Pi_*$ is the projector 
onto the eigenspace of $K$ 
with the eigenvalue $\lambda_*$. By definition of $\lambda_*$ one has 
\be
\label{max_eigenvector_1}
\psi\ne 0 \quad \mbox{and} \quad 
K |\psi\ra = \lambda_*|\psi\ra.
\ee
Fact~\ref{fact:seq} and the inclusion $K^m|0^n\ra \in LY(r)$ for each $m$ then imply
\be
\label{max_eigenvector_2}
\psi \in LY(r).
\ee
Below we assume that $K$ has rank at least two, so that $\lambda_2>0$.
In the remaining case when $K$ is a rank-one operator, one has $\lambda_1>0$, $\lambda_2=0$
and the bound Eq.~(\ref{relative_gap_bound}) is trivial.

Let $\phi$ be any eigenvector of $K$ with a positive eigenvalue $\lambda_\phi>0$
such that $\phi$ is linearly independent from $\psi$.
Define the spectral ratio 
\[
\alpha=\frac{\lambda_\phi}{\lambda_*}.
\]
Our goal is to prove that 
\be
\label{alpha_upper}
\alpha \le \frac1{r^2}.
\ee
This immediately implies that $\lambda_*=\lambda_1$ is the largest eigenvalue of $K$.
Indeed, if $\lambda_1>\lambda_*$ then one can choose $\phi=\psi_1$
orthogonal to $\psi$ obtaining $\alpha=\lambda_1/\lambda_*>1$ which contradicts Eq.~(\ref{alpha_upper}).
This also implies that $\lambda_*$ is non-degenerate. Indeed, otherwise
one can make $\alpha=1$ by choosing $\phi$
as any vector from the $\lambda_*$-eigenspace linearly independent of $\psi$,
which contradicts Eq.~(\ref{alpha_upper}).
Choosing $\phi$ as an eigenvector with the eigenvalue $\lambda_2$ proves the claimed bound
$\lambda_2/\lambda_1\le r^{-2}$. 

It remains to prove Eq.~(\ref{alpha_upper}).
Define the ratio function $g\, : \, \DD_r^n \to \CC$ such that 
\[
g(z)=\frac{f_{\phi}(z)}{f_{\psi}(z)}.
\]
Note that $g(z)$ is a holomorphic function on $\DD_r^n$ since
$f_{\psi}(z)$ has no zeros in $\DD_r^n$ due to Eq.~(\ref{max_eigenvector_2}).
The function $g(z)$ is nonconstant since
$g(z)\equiv c$ would force $|\phi\ra=c|\psi\ra$, contradicting 
linear independence.  
\begin{lemma}[\bf Squeezing]
For every $s\in(1/r,\,r]$ and every $v\in\DD_r^n$
with $g(v)\ne0$ one has
\be
\label{squeezing}
  \alpha\,g(v)\in g(\DD_s^n).
\ee
Here $g(\DD_s^n)=\{ g(z) \, : \,  z\in \DD_s^n\}$.
\end{lemma}
\begin{proof}
Pick any $v\in \DD_r^n$ with $g(v)\ne 0$.
Suppose Eq.~(\ref{squeezing}) is false, that is, $\alpha\,g(v) \notin g(\DD_s^n)$.
Let us show that this  leads to a contradiction.
Define
 \[
  \tau=-\frac1{\alpha g(v)} \quad \mbox{and}\quad
\ket{\varphi}=\ket{\psi}+\tau\ket{\phi}.
\]
We claim that $\varphi \in LY(s)$. Indeed,
for any  $\zeta \in \DD_s^n$ one has
\[
  f_\varphi(\zeta)=f_\psi(\zeta) + \tau f_\phi(\zeta)=
  f_{\psi}(\zeta)\bigl(1+\tau\,g(\zeta)\bigr)\ne0,
\]
since $f_{\psi}(\zeta)\ne0$ for $\zeta \in \DD_s^n\subseteq \DD_r^n$  and a zero of the second factor would give
$g(\zeta)=-1/\tau=\alpha g(v)$, which we assumed is not a value of $g$
on $\DD_s^n$.  Thus $\varphi \in LY(s)$.
Because $K\in LY(r)$ and  $rs>1$, the 
closure under contractions  of Fact~\ref{fact:contraction}
gives $K\varphi  \in LY(r)$. On the other hand,
\[
K\ket{\varphi}= \lambda_* |\psi\ra + \lambda_\phi \tau |\phi\ra=
\lambda_*\bigl(\ket{\psi}+\alpha \tau\ket{\phi}\bigr).
\]
Thus the generating polynomial of $K\varphi$ vanishes at $v$, 
\[
  f_{K\varphi}(v)
  =\lambda_*\,f_{\psi}(v)\bigl(1+\tau \alpha \,g(v)\bigr)
  =\lambda_*\,f_{\psi}(v)\,(1-1)=0.
\]
Since $v\in\DD_r^n$, this is  a contradiction to $K\varphi  \in LY(r)$.
  This proves the lemma.
\end{proof}
Fix real numbers $s$ and $t$  such that
\[
1/r<s<t<r.
\]
We shall eventually take the limits
 $s\downarrow1/r$ and $t\uparrow r$.
Given a compact set $M\subset \CC$,
let $\operatorname{diam}(M)=\max_{z,z' \in M} |z-z'|$ be the diameter of $M$.
Define
\[
\Delta_s=\operatorname{diam}(g(\overline{\DD_s^n})) \quad \mbox{and} \quad \Delta_t=\operatorname{diam}(g(\overline{\DD_t^n})).
\]
It follows directly from the definitions that
\be
\label{Delta_claim0}
0<\Delta_s \le \Delta_t <\infty.
\ee
Indeed, from  $s<t$ one gets $\overline{\DD_s^n} \subset \overline{\DD_t^n}$ and thus $\Delta_s\le \Delta_t$.
We have $\Delta_s>0$ since $g$ is a nonconstant holomorphic function on $\DD_r^n$
(by the identity theorem, $g$ is nonconstant on $\overline{\DD_s^n}$ as well).
Finally, $\Delta_t<\infty$ since $g$ is continuous on the compact set $\overline{\DD_t^n}\subset\DD_r^n$.

We claim that
\be
\label{Delta_claim1}
  \alpha\,\Delta_t \le \Delta_s.
\ee
Indeed, the zero set of $g$ has empty interior, so its
complement is dense in $\overline{\DD_t^n}$ and, by continuity, removing the value
$0$ from $g(\overline{\DD_t^n})$ does not change the diameter of  $g(\overline{\DD_t^n})$.
For any $z,z'\in \overline{\DD_t^n}$  with $g(z)\ne 0$ and $g(z')\ne 0$
the squeezing lemma gives $\alpha g(z)=g(v)$ and $\alpha g(z')=g(v')$
for some $v,v'\in \DD_s^n$. Thus
\[
\alpha |g(z)-g(z')| = |g(v)-g(v')| \le  \max_{v,v'\in \overline{\DD_s^n}} |g(v)-g(v')| = \Delta_s .
\]
Taking the supremum over all such $z,z'$ implies Eq.~(\ref{Delta_claim1}).

Next  we claim that
\be
\label{Delta_claim2}
\Delta_s\le \frac{s}{t}\Delta_t.
\ee
Indeed, we have
\[
\Delta_s=|g(z)-g(w)| \quad \mbox{for some $z,w \in \overline{\DD_s^n}$}.
\]
Let
\be
\label{rho_def}
\rho=\max\{\|z\|_\infty,\|w\|_\infty\}\le s.
\ee
Note that $\rho>0$ since otherwise  $\Delta_s=0$, contradicting Eq.~(\ref{Delta_claim0}).
Let $\DD_1=\{ \zeta \in \CC \, : \, |\zeta|<1\}$  be the open unit disk and $F\, : \, \DD_1\to \CC$ be a function defined as
\be
\label{F_from_g}
  F(\zeta)=\frac1{\Delta_t}\left( g(\rho^{-1} zt \zeta) - g(\rho^{-1} wt \zeta)\right)
\ee
for $\zeta \in \DD_1$.
We shall use the following result from complex analysis.
\begin{fact}[\bf Schwarz lemma]
Suppose $F \, : \, \DD_1 \to \CC$ is a holomorphic function such that $F(0)=0$ and $|F(\zeta)|\le 1$
for all $\zeta\in \DD_1$. Then $|F(\zeta)|\le |\zeta|$ for all $\zeta\in \DD_1$.
\end{fact}
The function $F(\zeta)$ defined in Eq.~(\ref{F_from_g}) is holomorphic on $\DD_1$ since
$\rho^{-1} zt \zeta \in \DD_t^n$, $\rho^{-1} wt \zeta \in \DD_t^n$ and $g$ is holomorphic on $\DD_t^n$.
Obviously,  $F(0)=0$. Furthermore, for any $\zeta \in \DD_1$ one has $|F(\zeta)|\le 1$ since
$g(\rho^{-1} zt \zeta)\in g(\DD_t^n)$ and $g(\rho^{-1} wt \zeta)\in g(\DD_t^n)$, while $\Delta_t$ is the diameter
of $g(\overline{\DD_t^n})$. Thus
\[
\Delta_s=|g(z)-g(w)|= \Delta_t |F(\rho t^{-1})| \le \Delta_t \rho t^{-1} \le \frac{s}{t} \Delta_t,
\]
where the first inequality follows from the Schwarz lemma and the second inequality follows from Eq.~(\ref{rho_def}).
This
proves Eq.~(\ref{Delta_claim2}).

Combining Eqs.~(\ref{Delta_claim1},\ref{Delta_claim2}) gives
$\alpha \Delta_t \le \Delta_s\le (s/t)\Delta_t$. From Eq.~(\ref{Delta_claim0}) one gets $\Delta_t\ne 0$.
Thus $\alpha \le s/t$ for all $s,t$ satisfying $1/r <s<t<r$.
Letting $s\downarrow1/r$ and $t\uparrow r$ yields
$\alpha\le r^{-2}$.
\end{proof}

\paragraph{Acknowledgements}
DG, YL, and BW acknowledge the support of the Natural Sciences and Engineering Research Council of Canada. We acknowledge the use of AI tools to develop proof strategies. 

\paragraph{Note added.} After completing this work, we became aware of independent work by
Rayudu and Takahashi~\cite{rayudu2026spectralgapleeyanghamiltonians}, which proves a gap lower bound $\mu/4$ for Lee--Yang Hamiltonians in a uniform
$Z$-field of strength $\mu>0$. Their proof uses decay of imaginary-time
correlations, whereas ours uses the relative spectral gap of
Lee--Yang operators and yields the tight lower bound $2\mu$ for
the Suzuki--Fisher class considered here. Both works independently
imply a polynomial-time quantum algorithm for the ground energy of
bipartite Quantum MaxCut.


\end{document}